\documentclass[twocolumn,final,3p,times]{elsarticle}
\usepackage[latin9]{inputenc}
\usepackage{amsmath}
\usepackage{amsthm}
\usepackage{amssymb}
\usepackage{graphicx}
\usepackage[unicode=true,
 bookmarks=false,
 breaklinks=false,pdfborder={0 0 1},backref=section,colorlinks=false]
 {hyperref}
\usepackage{breakurl}

\makeatletter
\theoremstyle{plain}
\newtheorem{thm}{\protect\theoremname}
  \theoremstyle{remark}
  \newtheorem{rem}[thm]{\protect\remarkname}
  \theoremstyle{definition}
  \newtheorem{defn}[thm]{\protect\definitionname}
  \theoremstyle{plain}
  \newtheorem{prop}[thm]{\protect\propositionname}
  \theoremstyle{plain}
  \newtheorem{lem}[thm]{\protect\lemmaname}
  \theoremstyle{plain}
  \newtheorem{cor}[thm]{\protect\corollaryname}

\journal{Systems and  Control Letters}

\makeatother

  \providecommand{\corollaryname}{Corollary}
  \providecommand{\definitionname}{Definition}
  \providecommand{\lemmaname}{Lemma}
  \providecommand{\propositionname}{Proposition}
  \providecommand{\remarkname}{Remark}
\providecommand{\theoremname}{Theorem}

\begin{document}

\begin{frontmatter}{}

\title{Distributed Kalman Filter in a Network of Linear Dynamical Systems}

\author[gua,arg]{Damián Marelli}

\ead{Damian.Marelli@newcastle.edu.au}

\author[new]{Mohsen Zamani}

\ead{Mohsen.Zamani@newcastle.edu.au}

\author[new,gua]{Minyue Fu}

\ead{Minyue.Fu@newcastle.edu.au}

\address[gua]{Department of Control Science and Engineering and State Key Laboratory
of Industrial Control Technology, Zhejiang University, 388 Yuhangtang
Road Hangzhou, Zhejiang Province, 310058, P. R. China.}

\address[arg]{French-Argentinean International Center for Information and Systems
Sciences, National Scientific and Technical Research Council, Ocampo
Esmeralda, Rosario 2000, Argentina.}

\address[new]{School of Electrical Engineering and Computer Science, University
of Newcastle, University Drive, Callaghan, NSW 2308, Australia.}
\begin{abstract}
This paper is concerned with the problem of distributed Kalman filtering
in a network of interconnected subsystems with distributed control
protocols. We consider networks, which can be either homogeneous or
heterogeneous, of linear time-invariant subsystems, given in the state-space
form. We propose a distributed Kalman filtering scheme for this setup.
The proposed method provides, at each node, an estimation of the
state parameter, only based on locally available measurements and
those from the neighbor nodes. The special feature of this method
is that it exploits the particular structure of the considered network
to obtain an estimate using only one prediction/update step at each
time step. We show that the estimate produced by the proposed method
asymptotically approaches that of the centralized Kalman filter, i.e.,
the optimal one with global knowledge of all network parameters, and
we are able to bound the convergence rate. Moreover, if the initial
states of all subsystems are mutually uncorrelated, the estimates
of these two schemes are identical at each time step. 
\end{abstract}

\end{frontmatter}{}

\section{Introduction}

There has been an increasing activity in the study of distributed
estimation in a network environment. This is due to its broad applications
in many areas, including formation control~\citet{subbotin2009distributed,lin2014distributed},
distributed sensor network~\citet{zhang2001stability} and cyber
security~\citet{teixeira2015secure,zamaniautomatic2015}. This paper
examines the problem of distributed estimation in a network of subsystems
represented by a finite dimensional state-space model. Our focus is
on the scenario where each subsystem obtains some noisy measurements,
and broadcasts them to its nearby subsystems, called \emph{neighbors}.
The neighbors exploit the received information, together with an estimate
of their internal states, to make a decision about their future states.
This sort of communication coupling arises in different applications.
For example, in control system security problems~\citet{teixeira2015secure},
distributed state estimation is required to calculate certain estimation
error residues for attack detection. Similarly, for formation control
\citet{lin2016graph,zheng2015distributed,lin2016necessary}, each
subsystem integrates measurements from its nearby subsystems, and
states of each subsystem need to be estimated for distributed control
design purposes. The main objective of this paper is to collectively
estimate the states of all subsystems within such a network. We will
propose a novel distributed version of the celebrated Kalman filter.

The current paper, in broad sense, belongs to the large body of literature
regarding distributed estimation. One can refer to~\citet{lopes2008diffusion,kar2012distributed,conejo2007optimization,gomez2011taxonomy,marellidistributed2015,olfati2005distributedkf,ugrinovskii2011robustestimate,ugrinovskii2013robustestimate,zamani2014minimumenergy,khan2008distributedkalmanfilter,olfati2009kalmanstability}
and the survey paper~\citet{ribeiro2010kalman}, as well as references
listed therein, for different variations of distributed estimation
methods among a group of subsystems within a network. A consensus
based Kalman filter was proposed in~\citet{olfati2005distributedkf}.
The author of~\citet{ugrinovskii2011robustestimate} utilized a linear
matrix inequality to minimize a $H_{\infty}$ index associated with
a consensus based estimator, which can be implemented locally. Some
of the results there were then extended to the case of switching topology
in~\citet{ugrinovskii2013robustestimate}. The same problem was solved
using the minimum energy filtering approach in~\citet{zamani2014minimumenergy}.
A common drawback of the state estimation methods described above
is that, being based on consensus, they require, in theory, an infinite
number of consensus iterations at each time step. This results in
computational and communication overload. To avoid this, in this paper
we exploit the network structure to achieve a distributed Kalman filter
method which requires only one prediction/update step at each time
step. Moreover, it is worthwhile noting that there is a major difference
between the above-mentioned works and the problem formulation in the
current paper. More precisely, in the former, the aim of each subsystem
is to estimate the aggregated state which is common to all subsystems.
In contrast, in the problem studied here, each subsystem is dedicated
to the estimation of its own internal state, which in general is different
from those of other subsystems. This allows the distributed estimation
algorithm to be scalable to networked systems with a large number
of subsystems where requiring each subsystem to estimate the aggregated
state is both computationally infeasible and practically unnecessary.

To show the effectiveness of the proposed algorithm, we compare our
method with the classical (centralized) Kalman filter, which is known
to be optimal (in the minimum error covariance sense). The classical
method requires the simultaneous knowledge of parameters and measurements
from all subsystems within the network to carry out the estimation.
In contrast, our proposed distributed estimation algorithm runs a
local Kalman filter at each subsystem, which only requires the knowledge
of local measurements and parameters, as well as measurements from
neighbor subsystems. Hence, it can be implemented in a fully distributed
fashion. We show that the state estimate, and its associated estimation
error covariance matrix, produced by the proposed distributed method
asymptotically converge to those produced by the centralized Kalman
filter. We provide bounds for the convergence of both the estimate
and the estimation error covariance matrix. A by-product of our result
is that, if the initial states of all subsystems are uncoupled (i.e.,
they are mutually uncorrelated), the estimates produced by our method
are identical to that of the centralized Kalman filter.

The rest of the paper is structured as follows. In Section~\ref{sec:Problem-setup}, 
we describe the network setup and its associated centralized Kalman
filter. In Section~\ref{sec:Distributed-Kalman-Filter}, we describe
the proposed distributed Kalman filter scheme. In Section~\ref{sec:Optimality-analysis}, 
we demonstrate the asymptotic equivalence between the proposed distributed
filter and the centralized one, and provide bounds for the convergence
of the estimates and their associated estimation error covariances.
Simulation results that support our theoretical claims are presented
in Section~\ref{sec:Simulations}. Finally, concluding remarks are
given in Section~\ref{sec:Conclusion}.

\section{System Description\label{sec:Problem-setup}}

In this paper we study networks of $I$ linear time-invariant subsystems.
Subsystem $i$ is represented by the following state-space model 
\begin{eqnarray}
x_{k+1}^{(i)} & = & A^{(i)}x_{k}^{(i)}+z_{k}^{(i)}+w_{k}^{(i)},\label{eq:sys1}\\
y_{k}^{(i)} & = & C^{(i)}x_{k}^{(i)}+v_{k}^{(i)}.\label{eq:sys2}
\end{eqnarray}
The subsystems are interconnected as follows 
\begin{equation}
z_{k}^{(i)}=\sum_{j\in\mathcal{N}_{i}}L^{(i,j)}y_{k}^{(j)},\label{eq:sys3}
\end{equation}
where $x_{k}^{(i)}\in\mathbb{R}^{n_{i}}$ is the state, $y_{k}^{(i)}\in\mathbb{R}^{p_{i}}$
the output, $w_{k}^{(i)}$ is an i.i.d Gaussian disturbance process
with $w_{k}^{(i)}\sim\mathcal{N}\left(0,Q_{i}\right)$, and $v_{k}^{(i)}$
is an i.i.d.~Gaussian measurement noise process with $v_{k}^{(i)}\sim\mathcal{N}\left(0,R_{i}\right)$.
We further suppose that $\mathcal{E}\left(w_{k}^{(i)}w_{k}^{(j)^{\top}}\right)=0$
and $\mathcal{E}\left(v_{k}^{(i)}v_{k}^{(j)^{\top}}\right)=0$, $\forall i\ne j$
and $\mathcal{E}\left(x_{k}^{(i)}w_{k}^{(j)^{\top}}\right)=0$, $\mathcal{E}\left(x_{k}^{(i)}v_{k}^{(j)^{\top}}\right)=0$
$\forall i,j$. We also denote the neighbor set of the subsystem $i$
by $\mathcal{N}_{i}=\left\{ j:L^{(i,j)}\neq0\right\} $.
\begin{rem}
We note in (\ref{eq:sys1})-(\ref{eq:sys2}) that the coupling between
neighboring subsystems is solely caused through the $z_{k}^{(i)}$
term in (\ref{eq:sys3}). The main motivation for considering such
coupling comes from distributed control where (\ref{eq:sys1}) represents
the model of an autonomous subsystem (or agent) with the $z_{k}^{(i)}$
being the control input and that (\ref{eq:sys3}) represents a distributed
control protocol which employs feedback only from neighboring measurements.
This type of distributed control is not only common for control of
multi-agent systems (see, for example,~\citep{lin2014distributed,lin2016graph,lin2016necessary,zheng2015distributed}),
but also realistic for large networked systems in the sense that only
neighbouring information is both easily accessible and most useful
for each subsystem. We emphasize that the distributed state estimation
problem arises for the networked system (\ref{eq:sys1})-(\ref{eq:sys3})
because of our allowance for measurement noises $v_{k}^{(i)}$ in
(\ref{eq:sys2}). This consideration is very important for applications
because measurement noises are unavoidable in practice. This also
sharply distinguishes our distributed control formulation from most
distributed control algorithms in the literature where perfect state
measurement is often implicitly assumed. 

We define $\xi_{k}^{\top}=\left[\left(\xi_{k}^{(1)}\right)^{\top},\cdots,\left(\xi_{k}^{(I)}\right)^{\top}\right]$
and $\Xi_{k}=\left\{ \xi_{1},\cdots,\xi_{k}\right\} $, where $(\xi,\Xi)$
stands for either $(x,X)$, $(y,Y)$, $(z,Z)$, $(w,W)$ or $(v,V)$;
moreover, we denote $\Upsilon=\mathrm{diag}\left\{ \Upsilon^{(1)},\cdots,\Upsilon^{(I)}\right\} $,
where $\Upsilon$ stands for either $A$, $B$, $C$, $Q$ or $R$,
and $L=\left[L^{(i,j)}:i,j=1,\cdots,I\right]$.
\end{rem}

Using the above notation, we let the initial state $x_{1}$ of all
subsystems have the joint distribution $x_{0}\sim\mathcal{N}\left(0,P\right)$.
We can also write the aggregated model of the whole network as 
\begin{eqnarray}
x_{k+1} & = & Ax_{k}+LCx_{k}+w_{k}+BLv_{k}\nonumber \\
 & = & \tilde{A}x_{k}+e_{k},\label{eq:ss1}\\
y_{k} & = & Cx_{k}+v_{k},\label{eq:ss2}
\end{eqnarray}
with 
\begin{equation}
\tilde{A}=A+LC\quad\text{ and }\quad e_{k}=w_{k}+Lv_{k}.\label{eq:Atilde}
\end{equation}
It then follows that 
\begin{equation}
\mathrm{cov}\left(\left[\begin{array}{c}
e_{k}\\
v_{k}
\end{array}\right]\left[\begin{array}{cc}
e_{k}^{\top} & v_{k}^{\top}\end{array}\right]\right)=\left[\begin{array}{cc}
\tilde{Q} & \tilde{S}\\
\tilde{S}^{\top} & R
\end{array}\right],\label{eq:correlated-noise}
\end{equation}
where $\tilde{Q}=Q+LRL^{\top}$ and $\tilde{S}=LR$.

\section{Centralized Kalman Filter}

Consider the standard (centralized) Kalman filter. For all $k,l\in\mathbb{N}$,
let

\begin{equation}
\begin{split}\hat{x}_{k|l} & \triangleq\mathcal{E}\left(x_{k}|Y_{l}\right),\\
\Sigma_{k|l} & \triangleq\mathcal{E}\left(\left[x_{k}-\hat{x}_{k|l}\right]\left[x_{k}-\hat{x}_{k\mid l}\right]^{\top}\right),
\end{split}
\label{eq:mllast}
\end{equation}
and $\Sigma_{0|0}=P$. Our aim in this subsection is to compute $\hat{x}_{k|k}$
in a standard centralized way. Notice that equation~(\ref{eq:correlated-noise})
implies that, in the aggregated system formed by~(\ref{eq:sys1})-(\ref{eq:sys2}),
the process noise $e_{k}$ and the measurement noise $v_{k}$ are
mutually correlated. Taking this into account, it follows from~\citep[S 5.5]{anderson-moore1979}
that the prediction and update steps of the (centralized ) Kalman
filter are given by: 
\begin{enumerate}
\item \textbf{Prediction:} 
\begin{equation}
\begin{split}\hat{x}_{k+1|k} & =\left(\tilde{A}-\tilde{S}R^{-1}C\right)\hat{x}_{k|k}+\tilde{S}R^{-1}y_{k}\\
 & =A\hat{x}_{k|k}+Ly_{k},
\end{split}
\label{eq:KP1-1}
\end{equation}
and 
\end{enumerate}
\begin{equation}
\begin{split}\Sigma_{k+1|k} & =\left(M-\tilde{S}R^{-1}C\right)\Sigma_{k|k}\left(M-\tilde{S}R^{-1}C\right)^{\top}\\
 & +\tilde{Q}-\tilde{S}R^{-1}\tilde{S}\\
 & =A\Sigma_{k|k}A^{\top}+Q.
\end{split}
\label{eq:KP2-1}
\end{equation}
\begin{enumerate}
\item \textbf{Update:} 
\begin{align}
\hat{x}_{k|k} & =\hat{x}_{k|k-1}+K_{k}\left(y_{k}-C\hat{x}_{k|k-1}\right),\label{eq:KU1-1}\\
\Sigma_{k|k} & =\left(I-K_{k}C\right)\Sigma_{k|k-1},\label{eq:KU2-1}
\end{align}
with 
\begin{equation}
K_{k}=\Sigma_{k|k-1}C^{\top}\left(C\Sigma_{k|k-1}C^{\top}+R\right)^{-1}.\label{eq:KG-1}
\end{equation}
\end{enumerate}

\section{Distributed Kalman Filter\label{sec:Distributed-Kalman-Filter}}

Consider the $i$-th subsystem~(\ref{eq:sys1})-(\ref{eq:sys2}).
Notice that, since the measurements $y_{k}^{(j)}$, $j\in\mathcal{N}_{i}$,
are known by the $i$-th subsystem, they can be treated as external
inputs. This observation leads us to the following intuitive approach
for a distributed Kalman filter scheme.

Let, for all $i=1,\cdots,I$ and $k,l\in\mathbb{N}$,

\begin{equation}
\begin{split}\hat{x}_{k|l}^{(i)} & \triangleq\mathcal{E}\left(x_{k}^{(i)}|y_{m}^{(j)};j\in\mathcal{N}_{i}\cup\{i\},m=1,\cdots,l\right),\\
\Sigma_{k|l}^{(i)} & \triangleq\mathcal{E}\left(\left[x_{k}^{(i)}-\hat{x}_{k|l}^{(i)}\right]\left[x_{k}^{(i)}-\hat{x}_{k|l}^{(i)}\right]^{\top}\right),
\end{split}
\label{eq:mllast-1}
\end{equation}
and $\Sigma_{0|0}^{(i)}=P^{(i)}$. Then, the prediction and update
steps for the proposed distributed Kalman filter are: 
\begin{enumerate}
\item \textbf{Prediction:} 
\begin{eqnarray}
\hat{x}_{k+1|k}^{(i)} & = & A^{(i)}\hat{x}_{k|k}^{(i)}+\sum_{j\in\mathcal{N}_{i}}L_{i,j}^{(i,j)}y_{k}^{(j)},\label{eq:DKP1}\\
\Sigma_{k+1|k}^{(i)} & = & A^{(i)}\Sigma_{k|k}^{(i)}A^{(i)^{\top}}+Q^{(i)}.\label{eq:DKP2}
\end{eqnarray}
\item \textbf{Update:} 
\begin{align}
\hat{x}_{k|k}^{(i)} & =\hat{x}_{k|k-1}^{(i)}+K_{k}^{(i)}\left(y_{k}^{(i)}-C^{(i)}\hat{x}_{k|k-1}^{(i)}\right),\label{eq:DKU1}\\
\Sigma_{k|k}^{(i)} & =\left(I-K_{k}^{(i)}C^{(i)}\right)\Sigma_{k|k-1}^{(i)},\label{eq:DKU2}
\end{align}
with 
\begin{equation}
K_{k}^{(i)}=\Sigma_{k|k-1}^{(i)}C^{(i)^{\top}}\left(C^{(i)}\Sigma_{k|k-1}^{(i)}C^{(i)^{\top}}+R^{(i)}\right)^{-1}.\label{eq:DKG}
\end{equation}
\end{enumerate}

\section{Optimality analysis\label{sec:Optimality-analysis}}

Since the distributed Kalman filter approach given in Section~\ref{sec:Distributed-Kalman-Filter}
is motivated by intuition, the question naturally arises as to which
extent it is optimal. In this section we address this question. To
this end, we define $\left(\hat{x}_{k|l}^{\star},\Sigma_{k|l}^{\star}\right)$,
where $\hat{x}_{k|l}^{\star\top}=\left[\hat{x}_{k|l}^{(i)\top}:i=1,\cdots,N\right]$
and $\Sigma_{k|l}^{\star}=\mathrm{diag}\left(\Sigma_{k|l}^{(i)}:i=1,\cdots,N\right)$,
to be the outcomes of distributed filter and $\left(\hat{x}_{k|l},\Sigma_{k|l}\right)$
to be those of centralized one. In Section~\ref{subsec:Convergence-of-covariance}, we show that the estimation error covariance of the distributed filter
$\Sigma_{k|k}^{\star}$ converges to that of the centralized one $\Sigma_{k|k}$,
and provide a bound for this convergence. In Section~\ref{subsec:Convergence-of-estimates}, we do the same for the convergence of $\hat{x}_{k|k}^{\star}$ to
$\hat{x}_{k|k}$.

\subsection{\label{subsec:Convergence-of-covariance} Convergence of $\Sigma_{k|k}^{\star}$
to $\Sigma_{k|k}$}

In this section, we show that the covariance matrices $\Sigma_{k|k}$
and $\Sigma_{k|k}^{\star}$ exponentially converge to each other,
and introduce a bound on $\left\Vert \Sigma_{k|k}-\Sigma_{k|k}^{\star}\right\Vert $.
To this end, we require the following definition from~\citep[Def 1.4]{bougerol1993kalman}. 
\begin{defn}
For $n\times n$ matrices $P,Q>0$, the Riemannian distance is defined
by 
\[
\delta\left(P,Q\right)=\sqrt{\sum_{k=1}^{n}\log^{2}\sigma_{k}\left(PQ^{-1}\right)},
\]
where $\sigma_{1}\left(X\right)\geq\cdots\geq\sigma_{n}\left(X\right)$
denote the singular values of matrix $X$. 

Several properties of the above definition, which we use to derive
our results, are given in the following proposition. 
\end{defn}

\begin{prop}
\label{prop:rim-dist}\citep[Proposition 6]{SuiMarelliFuBP} For $n\times n$
matrices $P,Q>0$, the following holds true: 
\begin{enumerate}
\item $\delta(P,P)=0$. 
\item $\delta\left(P^{-1},Q^{-1}\right)=\delta\left(Q,P\right)=\delta\left(P,Q\right).$ 
\item For any $m\times m$ matrix $W>0$ and $m\times n$ matrix $B$, we
have 
\[
\delta(W+BPB^{\top},W+BQB^{\top})\leq\frac{\alpha}{\alpha+\beta}\delta(P,Q),
\]
where $\alpha=\max\{\|BPB^{\top}\|,\|BQB^{\top}\|\}$ and $\beta=\sigma_{\min}\left(W\right)$. 
\item If $P>Q$, then $\left\Vert P-Q\right\Vert \leq\left(e^{\delta\left(P,Q\right)}-1\right)\left\Vert Q\right\Vert .$ 
\end{enumerate}
\end{prop}

The main result of this section is given in Theorem~\ref{thm:conv-Sigma}.
Its proof requires the following technical result.
\begin{lem}
\label{lem:matrix}Let $\Gamma_{k|l}=\Sigma_{k|l}^{-1}$ and $\Gamma_{k|l}^{\star}=\Sigma_{k|l}^{\star^{-1}}$.
Then 
\begin{eqnarray*}
\left\Vert \Sigma_{k|k}\right\Vert ,\left\Vert \Sigma_{k|k}^{\star}\right\Vert  & \leq & \sigma,\\
\left\Vert \Gamma_{k|k}\right\Vert ,\left\Vert \Gamma_{k|k}^{\star}\right\Vert  & \leq & \omega,
\end{eqnarray*}
and 
\begin{eqnarray}
\delta\left(\Sigma_{k|k},\Sigma_{k|k}^{\star}\right) & \leq & \upsilon^{k}\delta\left(P,P^{\star}\right),\label{eq:aux3}\\
\delta\left(\Gamma_{k|k},\Gamma_{k|k}^{\star}\right) & \leq & \upsilon^{k}\delta\left(P,P^{\star}\right),\label{eq:aux4}
\end{eqnarray}
where 
\begin{eqnarray}
\sigma & = & \max\left\{ \left\Vert P\right\Vert ,\left\Vert P^{\star}\right\Vert ,\left\Vert \bar{\Sigma}\right\Vert \right\} ,\label{eq:sigma}\\
\omega & = & \max\left\{ \left\Vert P^{-1}\right\Vert ,\left\Vert P^{\star^{-1}}\right\Vert ,\left\Vert \bar{\Sigma}^{-1}\right\Vert \right\} ,\label{eq:gamma}
\end{eqnarray}
with $P^{\star}$ denoting the diagonal matrix formed by the block
diagonal entries of the matrix $P$, 
\begin{eqnarray}
\upsilon & = & \upsilon_{1}\upsilon_{1},\qquad\upsilon_{2}=\frac{\sigma\left\Vert A\right\Vert ^{2}}{\sigma\left\Vert A\right\Vert ^{2}+\left\Vert Q^{-1}\right\Vert ^{-1}},\label{eq:upsilon}\\
\upsilon_{2} & = & \frac{\omega}{\omega+\left\Vert U^{-1}\right\Vert ^{-1}},\qquad U=C^{\top}R^{-1}C,\nonumber 
\end{eqnarray}
and $\bar{\Sigma}=\lim_{k\rightarrow\infty}\Sigma_{k|k}$.
\end{lem}

\begin{proof}
Let $\bar{\Sigma}^{\star}=\lim_{k\rightarrow\infty}\Sigma_{k|k}^{\star}$
and 
\begin{eqnarray}
\tilde{\sigma} & = & \max\left\{ \left\Vert P\right\Vert ,\left\Vert P^{\star}\right\Vert ,\left\Vert \bar{\Sigma}\right\Vert ,\left\Vert \bar{\Sigma}^{\star}\right\Vert \right\} ,\label{eq:sigma-alt}\\
\tilde{\omega} & = & \max\left\{ \left\Vert P^{-1}\right\Vert ,\left\Vert P^{\star^{-1}}\right\Vert ,\left\Vert \bar{\Sigma}^{-1}\right\Vert ,\left\Vert \bar{\Sigma}^{\star^{-1}}\right\Vert \right\} .\label{eq:gamma-alt}
\end{eqnarray}
We can then appeal to the fact that the Riccati equation is monotonic
\citet{bitmead1985monotonicity}, to conclude that, for all $k\in\mathbb{N}$,
\begin{eqnarray}
\left\Vert \Sigma_{k|k}\right\Vert  & \leq & \max\left\{ \left\Vert P\right\Vert ,\left\Vert \bar{\Sigma}\right\Vert \right\} \leq\tilde{\sigma},\label{eq:Sigma-bound}\\
\left\Vert \Sigma_{k|k}^{\star}\right\Vert  & \leq & \max\left\{ \left\Vert P^{\star}\right\Vert ,\left\Vert \bar{\Sigma}^{\star}\right\Vert \right\} \leq\tilde{\sigma},\label{eq:Sigma-star-bound}\\
\left\Vert \Gamma_{k|k}\right\Vert  & \leq & \max\left\{ \left\Vert P^{-1}\right\Vert ,\left\Vert \bar{\Sigma}^{-1}\right\Vert \right\} \leq\tilde{\omega},\label{eq:Gamma-bound}\\
\left\Vert \Gamma_{k|k}^{\star}\right\Vert  & \leq & \max\left\{ \left\Vert P^{\star^{-1}}\right\Vert ,\left\Vert \bar{\Sigma}^{\star^{-1}}\right\Vert \right\} \leq\tilde{\omega}.\label{eq:Gamma-star-bound}
\end{eqnarray}
Recall that 
\[
\Sigma_{k+1|k}=A\Sigma_{k|k}A^{\top}+Q.
\]
Also, from~\citep[p. 139]{anderson-moore1979}, we have 
\[
\Gamma_{k|k}=\Gamma_{k|k-1}+U.
\]
Clearly, similar relations hold for $\Sigma_{k|l}^{\star}$ and $\Gamma_{k|l}^{\star}$.
Then, it follows from Proposition~\ref{prop:rim-dist}-3 that, 
\begin{align}
\delta\left(\Sigma_{k+1|k},\Sigma_{k+1|k}^{\star}\right)= & \delta\left(A\Sigma_{k|k}A^{\top}+Q,A\Sigma_{k|k}^{\star}A^{\top}+Q\right)\nonumber \\
\leq & \tilde{\upsilon}_{1}\delta\left(\Sigma_{k|k},\Sigma_{k|k}^{\star}\right),\label{eq:aux1}\\
\delta\left(\Gamma_{k|k},\Gamma_{k|k}^{\star}\right)= & \delta\left(\Gamma_{k|k-1}+U,\Gamma_{k|k-1}^{\star}+U\right)\nonumber \\
\leq & \tilde{\upsilon}_{2}\delta\left(\Gamma_{k|k-1},\Gamma_{k|k-1}^{\star}\right),\label{eq:aux2}
\end{align}
with
\[
\tilde{\upsilon}_{1}=\frac{\tilde{\sigma}\left\Vert A\right\Vert ^{2}}{\tilde{\sigma}\left\Vert A\right\Vert ^{2}+\left\Vert Q^{-1}\right\Vert ^{-1}}\quad\text{ and }\quad\tilde{\upsilon}_{2}=\frac{\tilde{\omega}}{\tilde{\omega}+\left\Vert U^{-1}\right\Vert ^{-1}}.
\]
It then follows from~(\ref{eq:aux1})-(\ref{eq:aux2}) and Proposition~\ref{prop:rim-dist}-2,
that 
\begin{eqnarray*}
\delta\left(\Sigma_{k|k},\Sigma_{k|k}^{\star}\right) & \leq & \tilde{\upsilon}^{k}\delta\left(P,P^{\star}\right),\\
\delta\left(\Gamma_{k|k},\Gamma_{k|k}^{\star}\right) & \leq & \tilde{\upsilon}^{k}\delta\left(P,P^{\star}\right).
\end{eqnarray*}
with $\tilde{\upsilon}=\tilde{\upsilon}_{1}\tilde{\upsilon}_{2}$.
Finally, the above implies that $\bar{\Sigma}^{\star}=\bar{\Sigma}$.
Hence, the parameters $\tilde{\sigma}$ and $\tilde{\omega}$ given
in~(\ref{eq:sigma-alt})-(\ref{eq:gamma-alt}) are equivalent to
$\sigma$ and $\omega$ in~(\ref{eq:sigma})-(\ref{eq:gamma}), respectively,
and the result follows. 
\end{proof}
We now introduce the main result of the section, stating a bound on
$\left\Vert \Sigma_{k|k}-\Sigma_{k|k}^{\star}\right\Vert $. 
\begin{thm}
\label{thm:conv-Sigma} Let $\tilde{\Sigma}_{k|l}=\Sigma_{k|l}-\Sigma_{k|l}^{\star}$
and $\tilde{\Gamma}_{k|l}=\Gamma_{k|l}-\Gamma_{k|l}^{\star}$. Then
\[
\left\Vert \tilde{\Sigma}_{k|k}\right\Vert \leq\kappa\sigma\upsilon^{k}\quad\text{and}\quad\left\Vert \tilde{\Gamma}_{k|k}\right\Vert \leq\kappa\omega\upsilon^{k},
\]
where 
\[
\kappa=e^{\delta\left(P,P^{\star}\right)}-1.
\]
\end{thm}

\begin{proof}
Using~(\ref{eq:sigma})-(\ref{eq:gamma}), together with~(\ref{eq:aux3})-(\ref{eq:aux4}),
Proposition~\ref{prop:rim-dist}-4 and Lemma~\ref{lem:aux}, we
obtain 
\begin{eqnarray*}
\left\Vert \tilde{\Sigma}_{k|k}\right\Vert  & \leq & \left(e^{\upsilon^{k}\delta\left(P,P^{\star}\right)}-1\right)\left\Vert \Sigma_{k|k}\right\Vert \leq\kappa\upsilon^{k}\left\Vert \Sigma_{k|k}\right\Vert \\
 & \leq & \kappa\sigma\upsilon^{k},\\
\left\Vert \tilde{\Gamma}_{k|k}\right\Vert  & \leq & \kappa\upsilon^{k}\left\Vert \Gamma_{k|k}\right\Vert \leq\kappa\omega\upsilon^{k}.
\end{eqnarray*}
\end{proof}

\subsection{\label{subsec:Convergence-of-estimates}Convergence of $\hat{x}_{k|k}^{\star}$
to $\hat{x}_{k|k}$}

In this subsection, we study the convergence of state estimate $\hat{x}_{k|k}^{\star}$,
obtained through the distributed method, and that of the centralized
one $\hat{x}_{k|k}$. Moreover, we derive a bound on the error $\hat{x}_{k|k}$-$\hat{x}_{k|k}^{\star}$.
We start by introducing a number of lemmas which are instrumental
for establishing our main results. 
\begin{lem}
\label{lem:diff-dyn} Let $\tilde{x}_{k|l}=\hat{x}_{k|l}-\hat{x}_{k|l}^{\star}$.
Then 
\begin{equation}
\tilde{x}_{k+1|k}=H_{k}\tilde{x}_{k|k-1}+\xi_{k}.\label{eq:error-dynamics}
\end{equation}
where 
\begin{eqnarray*}
H_{k} & = & A\left(I-\Sigma_{k|k}U\right),\\
\xi_{k} & = & a_{k}+b_{k},\\
a_{k} & = & A\Sigma_{k|k}\tilde{\Gamma}_{k|k}\hat{x}_{k|k-1}^{\star},\\
b_{k} & = & A\tilde{\Sigma}_{k|k}\Gamma_{k|k}^{\star}\hat{x}_{k|k}^{\star}.
\end{eqnarray*}
\end{lem}

\begin{proof}
Let $\gamma_{k|l}=\Gamma_{k|l}\hat{x}_{k|l}$, $\gamma_{k|l}^{\star}=\Gamma_{k|l}^{\star}\hat{x}_{k|l}^{\star}$
and $\tilde{\gamma}_{k|l}=\gamma_{k|l}-\gamma_{k|l}^{\star}$. We
can easily obtain 
\[
\tilde{x}_{k+1|k}=A\tilde{x}_{k|k}.
\]
Also, from~\citep[p. 140]{anderson-moore1979}, we obtain
\[
\tilde{\gamma}_{k|k}=\tilde{\gamma}_{k|k-1}.
\]
Then it is easy to check that 
\begin{multline*}
\tilde{x}_{k|k}=\hat{x}_{k|k}-\hat{x}_{k|k}^{\star}=\Sigma_{k|k}\gamma_{k|k}-\Sigma_{k|k}^{\star}\gamma_{k|k}^{\star}\\
=\Sigma_{k|k}\tilde{\gamma}_{k|k}+\tilde{\Sigma}_{k|k}\gamma_{k|k}^{\star},
\end{multline*}
and 
\begin{multline*}
\tilde{\gamma}_{k|k-1}=\gamma_{k|k-1}-\gamma_{k|k-1}^{\star}\\
=\Gamma_{k|k-1}\hat{x}_{k|k-1}-\Gamma_{k|k-1}^{\star}\hat{x}_{k|k-1}^{\star}=\Gamma_{k|k-1}\tilde{x}_{k|k-1}+\tilde{\Gamma}_{k|k-1}\hat{x}_{k|k-1}^{\star}.
\end{multline*}
We then have 
\begin{multline*}
\tilde{x}_{k+1|k}=A\Sigma_{k|k}\tilde{\gamma}_{k|k-1}+A\tilde{\Sigma}_{k|k}\gamma_{k|k}^{\star}=A\Sigma_{k|k}\Gamma_{k|k-1}\tilde{x}_{k|k-1}+\xi_{k}\\
=A\Sigma_{k|k}\left(\Gamma_{k|k}-U\right)\tilde{x}_{k|k-1}+\xi_{k}=H_{k}\tilde{x}_{k|k-1}+\xi_{k}.
\end{multline*}
\end{proof}
\begin{lem}
\label{lem:conv-x} Let 
\begin{equation}
\Delta_{k}=\mathcal{E}\left(\tilde{x}_{k|k-1}\tilde{x}_{k|k-1}^{\top}\right).\label{eq:sigmad}
\end{equation}
\textup{ Then} 
\begin{equation}
\Delta_{k}\leq H_{k}\Delta_{k-1}H_{k}^{\top}+\lambda\upsilon^{k}I,\label{eq:aux5}
\end{equation}
where $I$ is the identity matrix, $\upsilon$ is defined in~(\ref{eq:upsilon}),
and 
\begin{equation}
\lambda\triangleq\sup_{k\in\mathbb{N}}\left(\zeta+2\sqrt{\zeta\left\Vert H_{k}\right\Vert ^{2}\left\Vert \Delta_{k-1}\right\Vert }\right)<\infty,\label{eq:lambda_bounded}
\end{equation}
with 
\begin{eqnarray}
\zeta & = & \left(\alpha+\beta\right)+2\sqrt{\alpha\beta},\nonumber \\
\alpha & = & \kappa^{2}\omega^{2}\sigma^{2}\left\Vert A\right\Vert ^{2}\left(\sigma\left\Vert A\right\Vert ^{2}+\left\Vert Q\right\Vert \right),\label{eq:parameters}\\
\beta & = & \kappa^{2}\omega^{2}\sigma^{3}\left\Vert A\right\Vert ^{2}.\nonumber 
\end{eqnarray}
\end{lem}

\begin{proof}
We split the argument in steps:

Step 1) From Lemmas~\ref{lem:diff-dyn} and~\ref{lem:PD-bound}
\begin{eqnarray*}
\left\Vert \mathcal{E}\left(\xi_{k}\xi_{k}^{\top}\right)\right\Vert  & \leq & \left\Vert \mathcal{E}\left(a_{k}a_{k}^{\top}\right)\right\Vert +\left\Vert \mathcal{E}\left(b_{k}b_{k}^{\top}\right)\right\Vert \\
 & + & 2\sqrt{\left\Vert \mathcal{E}\left(a_{k}a_{k}^{\top}\right)\right\Vert \left\Vert \mathcal{E}\left(b_{k}b_{k}^{\top}\right)\right\Vert }.
\end{eqnarray*}
Now, using Lemma~\ref{lem:matrix},
\begin{multline*}
\left\Vert \mathcal{E}\left(a_{k}a_{k}^{\top}\right)\right\Vert \leq\left\Vert A\right\Vert ^{2}\left\Vert \Sigma_{k|k}\right\Vert ^{2}\left\Vert \tilde{\Gamma}_{k|k}\right\Vert ^{2}\left\Vert \mathcal{E}\left(\hat{x}_{k|k-1}^{\star}\hat{x}_{k|k-1}^{\star^{\top}}\right)\right\Vert \\
\leq\kappa^{2}\omega^{2}\sigma^{2}\left\Vert A\right\Vert ^{2}\left\Vert \Sigma_{k|k-1}^{\star}\right\Vert \upsilon^{2k}\leq\alpha\upsilon^{2k},
\end{multline*}
and 
\begin{eqnarray*}
\left\Vert \mathcal{E}\left(b_{k}b_{k}^{\top}\right)\right\Vert  & \leq & \left\Vert A\right\Vert ^{2}\left\Vert \Gamma_{k|k}^{\star}\right\Vert ^{2}\left\Vert \tilde{\Sigma}_{k|k}\right\Vert ^{2}\left\Vert \mathcal{E}\left(\hat{x}_{k|k}^{\star}\hat{x}_{k|k}^{\star^{\top}}\right)\right\Vert \\
 & \leq & \kappa^{2}\omega^{2}\sigma^{2}\left\Vert A\right\Vert ^{2}\left\Vert \Sigma_{k|k}^{\star}\right\Vert \upsilon^{2k}\leq\beta\upsilon^{2k}.
\end{eqnarray*}
Then 
\[
\left\Vert \mathcal{E}\left(\xi_{k}\xi_{k}^{\top}\right)\right\Vert \leq\zeta\upsilon^{2k}.
\]

Step 2) From~(\ref{eq:error-dynamics}) and Lemma~\ref{lem:PD-bound},
we have 
\begin{multline*}
\Delta_{k}\leq H_{k}\Delta_{k-1}H_{k}{}^{\top}+\mathcal{E}\left(\xi_{k}\xi_{k}^{\top}\right)\\
+2\sqrt{\left\Vert H_{k}\Delta_{k-1}H_{k}\right\Vert \left\Vert \mathcal{E}\left(\xi_{k}^{\top}\xi_{k}\right)\right\Vert }I\leq F_{k}(\Delta_{k-1}),
\end{multline*}
with 
\[
F_{k}(X)=H_{k}XH_{k}^{\top}+\left(\zeta\upsilon^{k}+2\sqrt{\zeta\left\Vert H_{k}\right\Vert ^{2}\left\Vert X\right\Vert }\right)I\upsilon^{k}.
\]
Clearly, if $A>B$ then $F_{k}(A)>F_{k}(B)$. Also, there clearly
exists $\bar{k}$ and $\bar{\Delta}$ such that $F_{k}(\bar{\Delta})<\bar{\Delta}$,
for all $k\geq\bar{k}$. Hence, $\lim_{k\rightarrow\infty}\Delta(k)<\infty$,
and the result follows.
\end{proof}
The following result states a family of upper bounds on the norm of
the covariance matrix of $\tilde{x}_{k|l}$.
\begin{thm}
\label{thm:main} Consider\textup{ $\Delta_{k}$ as defined in \eqref{eq:sigmad}.}
Let $H_{k}=V_{k}J_{k}V_{k}{}^{-1}$ and $\bar{H}=\bar{V}\bar{J}\bar{V}^{-1}$
be the Jordan decompositions of $H_{k}$ and $\bar{H}$, respectively.
Then for every $\epsilon>1$, there exists $k_{\epsilon}\in\mathbb{N}$
such that 
\[
\left\Vert \Delta_{k}\right\Vert \leq A_{\epsilon}\psi_{\epsilon}^{k}+B_{\epsilon}\upsilon^{k},
\]
where 
\[
A_{\epsilon}=\frac{\lambda\psi_{\epsilon}\phi_{\epsilon}}{\psi_{\epsilon}-\upsilon},\qquad B_{\epsilon}=\frac{\lambda\upsilon\phi_{\epsilon}}{\upsilon-\psi_{\epsilon}}.
\]
and 
\begin{eqnarray}
\psi_{\epsilon} & = & \epsilon\rho\left(\bar{H}\right),\qquad\bar{H}=\lim_{k\rightarrow\infty}H_{k},\label{eq:lim}\\
\phi_{\epsilon} & = & \epsilon\left\Vert \bar{V}\right\Vert ^{2}\left\Vert \bar{V}^{-1}\right\Vert ^{2}\left(\frac{m_{\epsilon}}{\epsilon\rho\left(\bar{H}\right)}\right)^{2(k_{\epsilon}-1)},\nonumber \\
m_{\epsilon} & = & \max\left\{ 1,\left\Vert H_{1}\right\Vert ,\cdots,\left\Vert H_{k_{\epsilon}-1}\right\Vert \right\} .\nonumber 
\end{eqnarray}
\end{thm}

\begin{proof}
We split the argument in steps:

Step 1) Let 
\[
D_{k}=H_{k}D_{k-1}H_{k}{}^{\top}+\lambda I\upsilon^{k}.
\]
with $D_{1}=0$, From~(\ref{eq:aux5}), and since $\Delta_{1}=D_{1}=0$,
it follows that 
\begin{equation}
\Delta_{k}\leq D_{k}.\label{eq:bound}
\end{equation}

Step 2) Let 
\begin{eqnarray*}
\Pi_{l,k} & = & H_{k-1}H_{k-2}\cdots H_{l}\\
 & = & V_{k-1}J_{k-1}V_{k-1}{}^{-1}\cdots V_{l}J_{l}V_{l}{}^{-1}.
\end{eqnarray*}
From~(\ref{eq:lim}), there exists $k_{\epsilon}\in\mathbb{N}$ such
that, for all $k\geq k_{\epsilon}$, 
\begin{eqnarray*}
\left\Vert V_{k}\right\Vert  & \leq & \sqrt{\epsilon}\left\Vert \bar{V}\right\Vert ,\\
\left\Vert V_{k}^{-1}\right\Vert  & \leq & \sqrt{\epsilon}\left\Vert \bar{V}^{-1}\right\Vert ,\\
\left\Vert V_{k+1}{}^{-1}V_{k}\right\Vert  & \leq & \sqrt{\epsilon},\\
\left\Vert J_{k}\right\Vert  & \leq & \sqrt{\epsilon}\rho\left(\bar{H}\right).
\end{eqnarray*}
Then, for all $k\geq l$, 
\begin{eqnarray*}
\left\Vert \Pi_{l,k}\right\Vert  & \leq & m_{\epsilon}^{k_{\epsilon}-l}\left\Vert V_{k-1}\right\Vert \left\Vert J_{k-1}\right\Vert \left\Vert V_{k}{}^{-1}V_{k-1}\right\Vert \times\\
 &  & \cdots\times\left\Vert V_{k_{\epsilon}+1}{}^{-1}V_{k_{\epsilon}}\right\Vert \left\Vert J_{k_{\epsilon}}\right\Vert \left\Vert V_{k_{\epsilon}}\right\Vert \\
 & \leq & \sqrt{\epsilon}\left\Vert \bar{V}\right\Vert \left\Vert \bar{V}^{-1}\right\Vert m_{\epsilon}^{k_{\epsilon}-l}\left(\epsilon\rho\left(\bar{H}\right)\right)^{k-k_{\epsilon}}\\
 & = & \sqrt{\epsilon}\left\Vert \bar{V}\right\Vert \left\Vert \bar{V}^{-1}\right\Vert \left(\frac{m_{\epsilon}}{\epsilon\rho\left(\bar{H}\right)}\right)^{k_{\epsilon}-l}\left(\epsilon\rho\left(\bar{H}\right)\right)^{k-l}\\
 & \leq & \sqrt{\phi_{\epsilon}}\psi_{\epsilon}^{\frac{k-l}{2}}.
\end{eqnarray*}

Step 3) We have 
\[
D_{k+1}=\lambda\sum_{l=1}^{k}\upsilon^{l}\Pi_{l,k}\Pi_{l,k}^{\top}.
\]
Let $d_{k}=\left\Vert D_{k}\right\Vert $. Then 
\begin{multline*}
d_{k}\leq\lambda\sum_{l=1}^{k}\upsilon^{l}\left\Vert \Pi_{l,k}\Pi_{l,k}^{\top}\right\Vert \leq\lambda\phi_{\epsilon}\sum_{l=1}^{k}\upsilon^{l}\psi_{\epsilon}^{k-l}=\left(h_{k}\ast u_{k}\right),
\end{multline*}
with 
\[
h_{k}=\phi_{\epsilon}\psi_{\epsilon}^{k}\qquad\text{and}\qquad u_{k}=\lambda\upsilon^{k}.
\]
Taking $z$-transform we get 
\begin{multline*}
d(z)=h(z)u(z)=\frac{\lambda\phi_{\epsilon}}{\left(1-\psi_{\epsilon}z^{-1}\right)\left(1-\upsilon z^{-1}\right)}\\
=\frac{A_{\epsilon}}{1-\psi_{\epsilon}z^{-1}}+\frac{B_{\epsilon}}{1-\upsilon z^{-1}}.
\end{multline*}
Hence, 
\[
d_{k}=A_{\epsilon}\psi_{\epsilon}^{k}+B_{\epsilon}\upsilon^{k},
\]
and the result follows from the definition of $d_{k}$ and~(\ref{eq:bound}).
\end{proof}
Theorem~\ref{thm:main} states that the covariance of the difference
between $\hat{x}_{k|k-1}^{\star}$ and $\hat{x}_{k|k-1}$ is bounded
by two exponential terms. The term $B_{\epsilon}\upsilon^{k}$ is
due to the convergence of the Kalman gain $K_{k}^{\star}$ to $K_{k}$,
while the term $A_{\epsilon}\psi_{\epsilon}^{k}$ is due to the convergence
of the states given by the system dynamics. In order to use this result
to show the asymptotic convergence of $\hat{x}_{k|k-1}^{\star}$ to
$\hat{x}_{k|k-1}$, we need that $\upsilon<1$ and $\psi_{\epsilon}<1$,
for some $\epsilon>0$. While it is clear from~(\ref{eq:upsilon})
that the former is true, guaranteeing the latter is not that straightforward.
The following proposition addresses this issue.
\begin{prop}
If the pair $\left[A,C\right]$ is completely detectable and the pair
$\left[A,Q^{1/2}\right]$ is completely stabilizable, then $\rho\left(\bar{H}\right)<1,$
where $\rho(\bar{H})$ denotes the spectral radius of matrix $\bar{H}$.
\end{prop}

\begin{proof}
Let $K_{k}^{\star}=\mathrm{diag}\left(K_{k}^{(i)}:i=1,\cdots,N\right)$.
From Theorem~\ref{thm:conv-Sigma}, 
\[
\lim_{k\rightarrow\infty}K_{k}=\lim_{k\rightarrow\infty}K_{k}^{\star}\triangleq\bar{K}.
\]
Now, 
\begin{eqnarray*}
\hat{x}_{k+1|k} & = & A\left(I-K_{k}C\right)\hat{x}_{k|k-1}+\left(AK_{k}+L\right)y_{k},\\
\hat{x}_{k+1|k}^{\star} & = & A\left(I-K_{k}^{\star}C\right)\hat{x}_{k|k-1}^{\star}+\left(AK_{k}^{\star}+L\right)y_{k}.
\end{eqnarray*}
Hence, if we had that $K_{k}=K_{k}^{\star}=\bar{K}$, for all $k\in\mathbb{N}$,
then 
\[
\tilde{x}_{k+1|k}=A\left(I-\bar{K}C\right)\tilde{x}_{k|k-1}.
\]
However, under the same assumption, according to Lemma~\ref{lem:diff-dyn},
$\tilde{x}_{k+1|k}=\bar{H}\tilde{x}_{k|k-1}$. Hence, 
\[
\bar{H}=A\left(I-\bar{K}C\right).
\]
i.e., $\bar{H}$ equals the matrix that determines the asymptotic
dynamics of the centralized Kalman filter's estimation error. Then,
in view of the model~(\ref{eq:ss1})-(\ref{eq:ss2}), the result
follows from~\citep[S 4.4]{anderson-moore1979}.
\end{proof}

\subsection{The case when the initial covariance is block diagonal }

It turns out that, when the initial covariance matrix has a block
diagonal structure both estimation methods are completely identical.
This is summarized in the following corollary. 
\begin{cor}
Consider the network of subsystems~(\ref{eq:sys1})-(\ref{eq:sys2}).
If the matrix $P$ is block diagonal, then the distributed Kalman
filter scheme~(\ref{eq:DKP1})-(\ref{eq:DKG}) produces, for each
$i$, the same estimate as the centralized Kalman filter~(\ref{eq:KP1-1})-(\ref{eq:KG-1}). 
\end{cor}

\begin{proof}
Recall that matrices $A$, $Q$, $C$ and $R$ are all block diagonal.
It then follows from~(\ref{eq:KP2-1}) that, if $\Sigma_{k|k}$ is
block diagonal, so is $\Sigma_{k+1|k}$. One can easily check from~(\ref{eq:KU2-1})
and~(\ref{eq:KG-1}) that the same holds for $K_{k}$ and $\Sigma_{k|k}$
if $\Sigma_{k|k-1}$ is block diagonal. Since $\Sigma_{1|0}=P$ is
block diagonal, it follows that the matrices $\Sigma_{k|k-1}$ and
$\Sigma_{k|k}$ remain block diagonal for all $k$. Now, it is straightforward
to verify that~(\ref{eq:KP1-1})-(\ref{eq:KG-1}) become equivalent
to~(\ref{eq:DKP1})-(\ref{eq:DKG}), when $\Sigma_{k|k}$ and $\Sigma_{k|k-1}$
are block diagonal. Hence, the distributed and centralized Kalman
filters produce the same estimate and the result follows. 
\end{proof}

\section{Simulations\label{sec:Simulations}}

In this section, we present numerical results to compare the performance
of the proposed distributed scheme~(\ref{eq:DKP1})-(\ref{eq:DKG})
and the optimal centralized scheme~(\ref{eq:KP1-1})-(\ref{eq:KG-1}).
We assume that each subsystem is a single integrator with a pole at
$0.2$. The communication graph is undirected, as in Fig.~\ref{fig:fiveagents },
and the nonzero values of $L^{(i,j)}$ are set to $0.3$. Furthermore,
$v_{k}\sim\mathcal{N}\left(0,0.1I_{5}\right)$ and $w_{k}\sim\mathcal{N}\left(0,0.1I_{5}\right)$.

\begin{figure}
\centering{}\includegraphics[width=0.5\columnwidth]{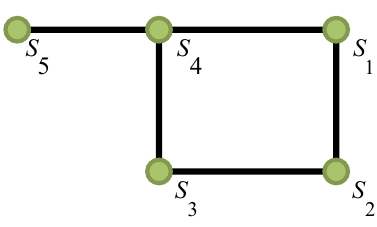}\caption{Communication graph of the five subsystems $S_{i}$, $i=1,2,\ldots,5$.
\label{fig:fiveagents }}
\end{figure}

We now compare the state estimation produced by the proposed distributed
scheme with that produced by the centralized one. To this end, we
set the initial conditions for both schemes to be the same, i.e.,
$\hat{x}_{0|0}=\hat{x}_{0|0}^{\star}=0$.

In the first simulation, the initial covariance matrix is chosen by
randomly generating a positive-definite matrix using $P=LL^{\top}+\epsilon_0 I_{5}$,
where $\epsilon_0=0.1$ and the entries $L\in\mathbb{R}^{5\times5}$
are drawn from the uniform distribution $\mathcal{U}(0,1)$. Fig.~\ref{fig:The-estimation-differences-dense}
shows the time evolution of the entries of the difference $\tilde{x}_{k|k-1}=\hat{x}_{k|k-1}^{\star}-\hat{x}_{k|k-1}$
between the estimation outcome $\hat{x}_{k|k-1}^{\star}$ of the distributed
Kalman filter and that $\hat{x}_{k|k-1}$ of the centralized one.
We see that how this difference asymptotically converges to zero.

\begin{figure}
\centering{}\includegraphics[width=0.7\columnwidth]{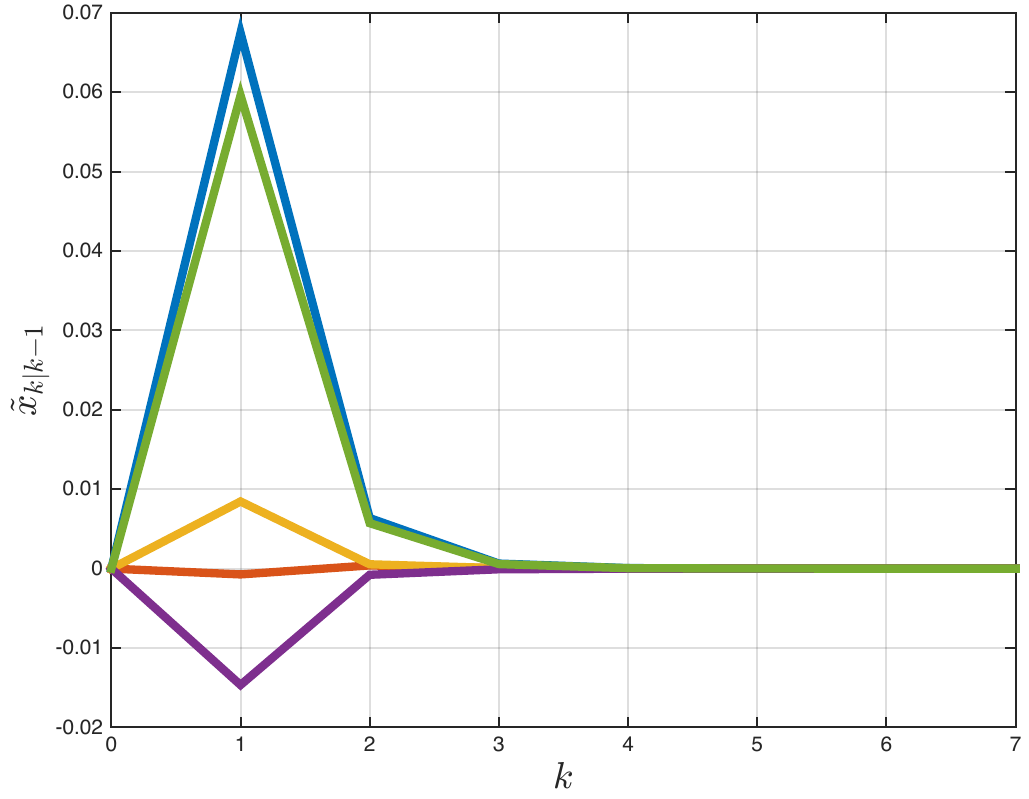}\caption{Difference between the estimated states obtained via both filtering
schemes.\label{fig:The-estimation-differences-dense}}
\end{figure}

In the second simulation we evaluate the same difference in the case
where the initial covariance matrix is block-diagonal. To this end,
we choose $P=\epsilon_1 I_{5}$, where $\epsilon_1$
is a random scaler drawn from the uniform distribution
$\mathcal{U}(0,1)$. The time evolution entries of the difference
$\tilde{x}_{k|k-1}$ are shown in Fig.~\ref{fig:The-estimation-differences}.
We see that these differences are very negligible, only due to numerical
errors. This confirms our claim that the estimates of both Kalman
filter schemes are the same when the matrix $P$ is block-diagonal.

\begin{figure}
\centering{}\includegraphics[width=0.7\columnwidth]{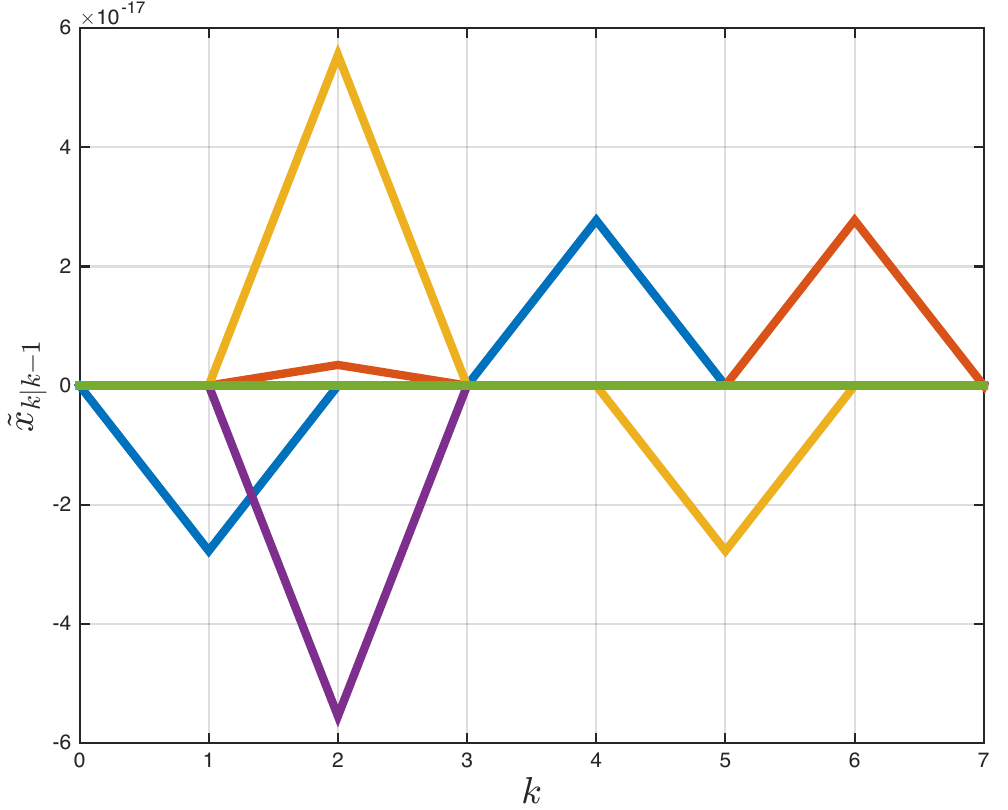}\caption{Difference between the estimated states obtained via both filtering
schemes, when $P$ is block-diagonal.\label{fig:The-estimation-differences}}
\end{figure}

\section{Conclusion\label{sec:Conclusion}}

We studied the distributed Kalman filter problem in a network of linear
time-invariant subsystems. We proposed a distributed Kalman filter
scheme, which uses only local measurements, and we studied the extent
to which this scheme approximates the centralized (i.e., optimal)
Kalman filter. It turned out that the covariance matrix associated
with the initial value of the state vector plays an important role.
We showed that if this matrix is block diagonal, the proposed distributed
scheme is optimal. Moreover, if that condition is dropped, the estimation
error covariances, and the associated estimates, obtained through
these two approaches approximate each other exponentially fast. We
also established proper bounds on error between estimates and its
covariance matrix.
\vspace{-4mm}
\appendix

\section{Some lemmas}
\begin{lem}
\label{lem:aux}\citep[Lemma 25]{SuiMarelliFuBP} For every $x\in\mathbb{R}$
and $0\leq y\leq1$, $e^{xy}-1\leq\left(e^{x}-1\right)y$.
\end{lem}

\begin{lem}
\label{lem:PD-bound}\citep[Lemma 26]{SuiMarelliFuBP} If $\left[\begin{array}{cc}
A & B^{\top}\\
B & C
\end{array}\right]\geq0$, then $\left\Vert B\right\Vert \leq\sqrt{\left\Vert A\right\Vert \left\Vert C\right\Vert }$. 
\end{lem}

\vspace{-4mm}
\bibliographystyle{plainnat}
\bibliography{mohsenbib2}

\begin{thebibliography}{24}
\providecommand{\natexlab}[1]{#1}
\providecommand{\url}[1]{\texttt{#1}}
\expandafter\ifx\csname urlstyle\endcsname\relax
  \providecommand{\doi}[1]{doi: #1}\else
  \providecommand{\doi}{doi: \begingroup \urlstyle{rm}\Url}\fi

\bibitem[Anderson and Moore(1979)]{anderson-moore1979}
B~.D.~O. Anderson and J.~B. Moore.
\newblock \emph{Optimal Filtering}.
\newblock Englewood Cliffs, NJ: Prentice-Hall, 1979.

\bibitem[Bitmead et~al.(1985)Bitmead, Gevers, Petersen, and
  Kaye]{bitmead1985monotonicity}
R.~R. Bitmead, M.~R. Gevers, I.~R. Petersen, and R.~J. Kaye.
\newblock Monotonicity and stabilizability-properties of solutions of the
  riccati difference equation: Propositions, lemmas, theorems, fallacious
  conjectures and counterexamples.
\newblock \emph{Systems \& Control Letters}, 5\penalty0 (5):\penalty0 309--315,
  1985.

\bibitem[Bougerol(1993)]{bougerol1993kalman}
P.~Bougerol.
\newblock Kalman filtering with random coefficients and contractions.
\newblock \emph{SIAM Journal on Control and Optimization}, 31\penalty0
  (4):\penalty0 942--959, 1993.

\bibitem[Conejo et~al.(2007)Conejo, la~Torre, and
  Canas]{conejo2007optimization}
A.~J. Conejo, S.~De la~Torre, and M.~Canas.
\newblock An optimization approach to multiarea state estimation.
\newblock \emph{IEEE Transactions on Power Systems}, 22\penalty0 (1):\penalty0
  213--221, 2007.

\bibitem[G{\'o}mez-Exp{\'o}sito et~al.(2011)G{\'o}mez-Exp{\'o}sito,
  la~Villa~Ja{\'e}n, G{\'o}mez-Quiles, P.~Rousseaux, and
  Cutsem]{gomez2011taxonomy}
A.~G{\'o}mez-Exp{\'o}sito, A.~De la~Villa~Ja{\'e}n, C.~G{\'o}mez-Quiles,
  Patricia P.~Rousseaux, and T.~Van Cutsem.
\newblock A taxonomy of multi-area state estimation methods.
\newblock \emph{Electric Power Systems Research}, 81\penalty0 (4):\penalty0
  1060--1069, 2011.

\bibitem[Kar et~al.(2012)Kar, Moura, and Ramanan]{kar2012distributed}
S.~Kar, J.~M.~F Moura, and K.~Ramanan.
\newblock Distributed parameter estimation in sensor networks: Nonlinear
  observation models and imperfect communication.
\newblock \emph{IEEE Transactions on Information Theory}, 58\penalty0
  (6):\penalty0 3575--3605, 2012.

\bibitem[Khan and Moura(2008)]{khan2008distributedkalmanfilter}
U.~A. Khan and J.~M.~F. Moura.
\newblock Distributing the kalman filter for large-scale systems.
\newblock \emph{IEEE Transactions on Signal Processing}, 56\penalty0
  (10):\penalty0 4919--4935, 2008.

\bibitem[Lin et~al.(2014)Lin, Wang, Han, and M.Fu]{lin2014distributed}
Z.~Lin, L.~Wang, Z.~Han, and M.Fu.
\newblock Distributed formation control of multi-agent systems using complex
  {L}aplacian.
\newblock \emph{IEEE Transactions Automatic Control}, 59\penalty0 (7):\penalty0
  1765--1777, 2014.

\bibitem[Lin et~al.(2016{\natexlab{a}})Lin, Wang, Chen, Fu, and
  Han]{lin2016necessary}
Z.~Lin, L.~Wang, Z.~Chen, M.~Fu, and Z.~Han.
\newblock Necessary and sufficient graphical conditions for affine formation
  control.
\newblock \emph{IEEE Transactions Automatic Control}, 61\penalty0
  (10):\penalty0 2877--2891, 2016{\natexlab{a}}.

\bibitem[Lin et~al.(2016{\natexlab{b}})Lin, Wang, Han, and Fu]{lin2016graph}
Z.~Lin, L.~Wang, Z.~Han, and M.~Fu.
\newblock A graph laplacian approach to coordinate-free formation stabilization
  for directed networks.
\newblock \emph{IEEE Transactions Automatic Control}, 61\penalty0 (5):\penalty0
  1269--1280, 2016{\natexlab{b}}.

\bibitem[Lopes and Ali(2008)]{lopes2008diffusion}
C.~G. Lopes and A.~H. Ali.
\newblock Diffusion least-mean squares over adaptive networks: Formulation and
  performance analysis.
\newblock \emph{IEEE T. Signal Proces}, 56\penalty0 (7):\penalty0 3122--3136,
  2008.

\bibitem[Marelli and Fu(2015)]{marellidistributed2015}
D.~Marelli and M.~Fu.
\newblock Distributed weighted linear least squares estimation with fast
  convergence in large-scale systems.
\newblock \emph{Automatica}, 51:\penalty0 27--39, 2015.

\bibitem[Olfati-Saber(2005)]{olfati2005distributedkf}
R.~Olfati-Saber.
\newblock Distributed kalman filter with embedded consensus filters.
\newblock In \emph{IEEE Conference on Decision and Control}, pages 8179--8184,
  2005.

\bibitem[Olfati-Saber(2009)]{olfati2009kalmanstability}
R.~Olfati-Saber.
\newblock Kalman-consensus filter : Optimality, stability, and performance.
\newblock In \emph{IEEE Conference on Decision and Control}, pages 7036--7042,
  2009.

\bibitem[Ribeiro et~al.(2010)Ribeiro, Schizas, Roumeliotis, and
  Giannakis]{ribeiro2010kalman}
A.~Ribeiro, I.~Schizas, S.~Roumeliotis, and G.~Giannakis.
\newblock Kalman filtering in wireless sensor networks.
\newblock \emph{IEEE Control Systems}, 30\penalty0 (2):\penalty0 66--86, 2010.

\bibitem[Subbotin and Smith(2009)]{subbotin2009distributed}
M.~Subbotin and R.~Smith.
\newblock Design of distributed decentralized estimators for formations with
  fixed and stochastic communication topologies.
\newblock \emph{Automatica}, 45\penalty0 (11):\penalty0 2491 -- 2501, 2009.

\bibitem[Sui et~al.()Sui, Marelli, and Fu]{SuiMarelliFuBP}
T.~Sui, D.~Marelli, and M.~Fu.
\newblock Accuracy analysis for distributed weighted least-squares estimation
  in finite steps and loopy networks.
\newblock submitted to Automatica.
\newblock URL
  \url{http://www.cifasis-conicet.gov.ar/marelli/DWLS_accuracy.pdf}.

\bibitem[Teixeira et~al.(2015)Teixeira, Shames, Sandberg, and
  Johansson]{teixeira2015secure}
A.~Teixeira, I.~Shames, H.~Sandberg, and K.~H. Johansson.
\newblock A secure control framework for resource-limited adversaries.
\newblock \emph{Automatica}, 51:\penalty0 135--148, 2015.

\bibitem[Ugrinovskii(2011)]{ugrinovskii2011robustestimate}
V.~Ugrinovskii.
\newblock Distributed robust filtering with consensus of estimates.
\newblock \emph{Automatica}, 47\penalty0 (1):\penalty0 1 -- 13, 2011.

\bibitem[Ugrinovskii(2013)]{ugrinovskii2013robustestimate}
V.~Ugrinovskii.
\newblock Distributed robust estimation over randomly switching networks using
  consensus.
\newblock \emph{Automatica}, 49\penalty0 (1):\penalty0 160 -- 168, 2013.

\bibitem[Zamani and Ugrinovskii(2014)]{zamani2014minimumenergy}
M.~Zamani and V.~Ugrinovskii.
\newblock Minimum-energy distributed filtering.
\newblock In \emph{IEEE Conference on Decision and Control}, pages 3370--3375,
  2014.

\bibitem[Zamani et~al.(2015)Zamani, Helmke, and Anderson]{zamaniautomatic2015}
M.~Zamani, U.~Helmke, and B.~D.~O. Anderson.
\newblock Zeros of networked systems with time-invariant interconnections,.
\newblock \emph{Automatica}, pages 97--105, 2015.

\bibitem[Zhang et~al.(2001)Zhang, Branicky, and Phillips]{zhang2001stability}
W.~Zhang, M.~S. Branicky, and S.~M. Phillips.
\newblock Stability of networked control systems.
\newblock \emph{IEEE Control Systems}, 21\penalty0 (1):\penalty0 84--99, 2001.

\bibitem[Zheng et~al.(2015)Zheng, Lin, Fu, and Sun]{zheng2015distributed}
R.~Zheng, Z.~Lin, M.~Fu, and D.~Sun.
\newblock Distributed control for uniform circumnavigation of ring-coupled
  unicycles.
\newblock \emph{Automatica}, 53:\penalty0 23--29, 2015.

\end{thebibliography}

\end{document}